\documentclass[letterpaper, 10 pt, conference]{ieeeconf}  % Comment this line out if you need a4paper

\IEEEoverridecommandlockouts                              % This command is only needed if 
\usepackage{graphicx}
\usepackage{amsmath}
 
\usepackage{amsthm}
\usepackage{amssymb}
\usepackage{amsfonts}  
\usepackage{stmaryrd}
\usepackage{multicol}
\usepackage{cuted}
\usepackage{multirow}
\usepackage{caption, subcaption}
\usepackage{dsfont}
\usepackage{xcolor}

\usepackage{tikz}
\usepackage{pgfplots}
\pgfplotsset{compat=newest}
\usetikzlibrary{plotmarks}
\usetikzlibrary{arrows.meta}
\usepgfplotslibrary{patchplots}
\usepackage{grffile}

\newtheorem{theorem}{Theorem}
\newtheorem{lemma}{Lemma}
\newtheorem{proposition}{Proposition}
\newtheorem{definition}{Definition}
\newtheorem{assumption}{Assumption}

\newcommand{\n}[1]{\ensuremath{n_\mathrm{#1}}}
\newcommand{\R}[1]{\ensuremath{\mathbb{R}^{#1}}}
\newcommand{\E}{\ensuremath{\mathbb{E}}}
\newcommand{\Tr}{\ensuremath{\mathrm{Tr}}}
\newcommand{\bbP}{\ensuremath{\mathbb{P}}}
\newcommand{\1}[1]{\ensuremath{\mathds{1}_{#1}}}
\newcommand{\bbF}{\ensuremath{\mathbb{F}}}

\title{\LARGE \bf
Probabilistically Input-to-State Stable Stochastic \\ Model Predictive Control
}

\author{Maik Pfefferkorn and Rolf Findeisen% <-this % stops a space
\thanks{M. Pfefferkorn and R. Findeisen are with the Control and Cyper-Physical Systems Laboratory, Technical University of Darmstadt, Germany
        {\tt\small \{maik.pfefferkorn, rolf.findeisen\}@iat.tu-darmstadt.de}}%
}

\begin{document}

\maketitle
\thispagestyle{empty}
\pagestyle{empty}

%%%%%%%%%%%%%%%%%%%%%%%%%%%%%%%%%%%%%%%%%%%%%%%%%%%%%%%%%%%%%%%%%%%%%%%%%%%%%%%%
\begin{abstract}
    Employing model predictive control to systems with unbounded, stochastic disturbances poses the challenge of guaranteeing safety, i.e., repeated feasibility and stability of the closed-loop system.
    Especially, there are no strict repeated feasibility guarantees for standard stochastic MPC formulations.
    Thus, traditional stability proofs are not straightforwardly applicable.
    We exploit the concept of input-to-state stability in probability and outline how it can be used to provide stability guarantees, circumventing the requirement for strict repeated feasibility guarantees.
    Loss of feasibility is captured by a back-up controller, which is explicitly taken into account in the stability analysis.
    We illustrate our findings using a numeric example.
\end{abstract}

%%%%%%%%%%%%%%%%%%%%%%%%%%%%%%%%%%%%%%%%%%%%%%%%%%%%%%%%%%%%%%%%%%%%%%%%%%%%%%%%

\section{Introduction}\label{sec:introduction}

Model predictive control (MPC)  \cite{Rawlings2019} is an advanced, widely used feedback control approach \cite{Schwenzer2021, Lucia2016}.
MPC relies on a prediction model of the system to determine control actions such that the future system behavior is optimized with respect to a task objective.
Standard nominal MPC formulations assume perfect knowledge of the system, which is often unrealistic in practice.
Furthermore, in the presence of uncertainty, as is usually the case in applications, the performance of nominal MPC schemes can significantly degrade.

To remedy from this situation, robust and stochastic MPC schemes can be employed, both exploiting high-level information about the nature of the uncertainty \cite{Rawlings2019}.
In particular, robust MPC \cite{Rakovic2012} is suited when the uncertainty is bounded, and control actions are determined such that they remain safe even for worst-case uncertainty realizations, irrespective of their probability to occur.
While this leads to rigorous and theoretically sound safety and stability guarantees, robust MPC is inherently conservative.
In contrast to that, stochastic MPC is suited for random disturbances.
In such case, worst-case-based constraints as used in robust MPC are relaxed to chance constraints, allowing for constraint violations if they occur with sufficiently low probability and thereby reducing conservatism.
% In consequence, stochastic MPC enables to optimize the average performance of the system at an acceptable risk level of constraint violations, which decreases conservatism compared to robust MPC.
Stochastic MPC has gained high interest in recent years for linear systems, e.g., \cite{Hewing2020b, Lorenzen2017, Kogel2021, Schluter2022, Brudigam2022, Schluter2023} and references therein, and nonlinear systems, e.g., \cite{Munoz2021, Mesbah2014, McAllsiter2023} and references therein.
However, its challenges lie in finding computationally tractable problem formulations and deriving rigorous safety and stability guarantees.

For systems with additive stochastic uncertainty, computationally tractable stochastic MPC formulations usually use a decomposition of the uncertain system into a deterministic, nominal part and a stochastic part involving the uncertainty.
The nominal dynamics are then used as prediction model and chance constraints are addressed via tightened constraints on the nominal system exploiting probabilistic reachable sets of the uncertainty dynamics \cite{Hewing2020a}. 

Obtaining safety, i.e., repeated feasibility, and stability guarantees for stochastic MPC is more challenging due to the stochastic nature of the disturbance and the MPC formulation. %, it is not possible to obtain strict repeated feasibility guarantees.
% In particular, the probability that constraints are violated at the same future time point is not the same at all time points prior to that \textcolor{red}{cite!Lorenzen and references therein}.
In particular, a feasible solution at the current time point is not guaranteed to remain feasible when executing it, and hence, cannot be used to construct a feasible candidate solution for the next time step, as usually done in MPC \cite{Lorenzen2017, Kouvaritakis2010}.
Rigorous and strict repeated feasibility guarantees can be obtained using a mixed robust-stochastic constraint tightening but requires boundedness of the disturbance \cite{Lorenzen2017, Munoz2021, McAllsiter2023, Kouvaritakis2010, Mayne2019}, thereby introducing conservatism as argued in \cite{Lorenzen2017}.
However, those strict repeated feasibility guarantees enable to prove stochastic stability.
For unbounded disturbances, strict repeated feasibility guarantees can be recovered by alternative MPC formulations that do not rely on initializing the prediction model by the measured uncertain state.
To this end, indirect feedback was proposed in \cite{Hewing2020b}, where the measured uncertain state is exploited for cost function evaluation and chance constraints are enforced using predictions conditioned on initial time instead of the most recent measurement. % and the nominal prediction model is initialized by the nominal one-step prediction of the previous horizon.
Chance constraint satisfaction in closed-loop is shown and a decrease condition for the optimal value function is given, providing a stochastic stability notion.
A similar approach was recently presented in \cite{Schluter2022, Schluter2023}, where the initial condition of each horizon is optimized over such that the MPC remains feasible at all times.
Correspondingly, chance constraint satisfaction in closed-loop and stability of the proposed MPC scheme are shown.
In contrast to that, the authors in \cite{Primbs2009} circumvent the requirement of repeated feasibility and prove stochastic stability for a stochastic MPC formulation in the case of unbounded disturbances.
Alternatively, for stochastic MPC formulations that lack strict repeated feasibility guarantees, back-up controllers can be employed to recover safety and stability when the MPC becomes infeasible, e.g., \cite{Brudigam2022, Paulson2017}.

We build upon the results in \cite{Primbs2009} and exploit the concept of input-to-state stability in probability (ISSp) \cite{Culbertson2023} to derive stability guarantees for linear stochastic MPC.
ISSp extends input-to-state stability (ISS), which is a widely used stability notion for systems subject to bounded disturbances, to systems subject to unbounded, stochastic disturbances.
Intuitively speaking, ISSp means that the ISS property holds with certain probability for finite horizons, indicating that the system converges in probability and for finite horizons to a neighborhood of the nominal system's equilibrium point.
We introduce a stochastic MPC formulation and its computationally tractable reformulation using probabilistic reachable sets, and propose a control policy that uses the MPC whenever feasible and exploits a back-up controller otherwise.
We show ISSp for the uncertain system under both the stochastic MPC and the back-up controller using Lyapunov-like arguments.
Combining those results, we conclude ISSp of the uncertain system under the proposed control policy.
Notably, the proposed ISSp guarantees do not require strict repeated feasibility of the MPC, rendering the results valid for a large class of stochastic MPC formulations where the uncertainty is unbounded or the bounds are unknown.
We illustrate ISSp according to our results using a numeric example.

The remainder of the article is organized as follows.
We introduce the considered stochastic MPC formulation in Section \ref{sec:smpc}.
Thereafter, stability of the SMPC formulation is investigated in Section \ref{sec:issp}, including the presentation of the main results.
We illustrate our findings using numeric examples in Section \ref{sec:simulation} before concluding in Section \ref{sec:conclusions}.

\noindent \textbf{Notation:} 
For a random vector $X$, we denote the expected value vector by $\mathbb{E}[X]$, where $\mathbb{E}[\cdot]$ is the expectation operator.
Further, $\bbP[X \in \mathcal{X}]$ is the probability of $X$ on a set $\mathcal{X}$ and $\mathds{1}_\mathcal{X}\{X=x\}$ denotes the indicator function on $\mathcal{X}$.
The unit matrix of dimension $n$ is denoted by $\mathbb{I}_n$, where the subscript is dropped when the dimension is clear from the context.
Given two real, symmetric matrices $C,D \in \R{n \times n}$, we denote the Loewner order by $C \preceq_L D$, indicating that $D - C$ is positive semi-definite, i.e., $D - C \succeq 0$.
For $a,b \in \mathbb{N}, a \leq b$, we denote by $\mathcal{I}_{a:b} := \{z \in \mathbb{N} \mid a \leq z \leq b \}$ the set of natural numbers between $a$ and $b$.
Given a compact set $\mathcal{W}$, we denote its interior by $\mathrm{int}(\mathcal{W})$ and its boundary by $\partial \mathcal{W}$.
For two sets $A$ and $B$, we denote their Pontryagin difference by $A \ominus B := \{ a \in A \mid \forall b \in B: a + b \in A \}$.
\section{Stochastic Model Predictive Control}\label{sec:smpc}

% \subsection{Problem Formulation}

We consider stochastic linear time-invariant systems
\begin{equation}\label{eq:system}
    x_{k+1} = A x_k + B u_k + w_k,~ x_0 = \bar{x},
\end{equation}
where $x_k \in \R{\n{x}}$, $u_k \in \R{\n{u}}$ and $w_k \in \R{\n{x}}$ denote the state, control input and process noise, respectively, at time point $k \in \mathbb{N}$.
The matrices $A \in \R{\n{x} \times \n{x}}$ and $B \in \R{\n{x} \times \n{u}}$ are assumed to be known.

\begin{assumption}\label{as:disturbance}
    The process noise $w_k \sim \mathcal{P}^w$ is a zero-mean random vector with covariance matrix $\Sigma^w = \E[w_k w_k^\top], \Sigma^w \preceq_L \nu \mathbb{I}$ for some $\nu \in \R{}_+, \nu < \infty$. The distribution $\mathcal{P}^w$ belongs to the moment-based ambiguity set $\mathcal{A}^w = \{ \mathcal{P}^w \mid \E[w_k] = 0, \E[w_k w_k^\top] = \Sigma^w$ \} and $w_k$ is independent and identically distributed (i.i.d.) over time.
\end{assumption}

\noindent In consequence of Assumption \ref{as:disturbance}, $x_{k+1} \sim \mathcal{P}_{k+1}^x$ becomes a random vector with $\mathcal{P}_{k+1}^x \in \mathcal{A}^x_{k+1} = \{ \mathcal{P}^x_{k+1} \mid \E[x_{k+1}] = z_{k+1}, \E \big[ (x_{k+1}-z_{k+1})(x_{k+1}-z_{k+1})^\top \big] = \Sigma^x_{k+1} \}$.
Therein, $z_{k+1}$ denotes the state mean with nominal dynamics
\begin{align}
    z_{k+1} & = A z_k + B u_k,~ z_0 = \bar{x}, \label{eq:nominal_system}
\end{align}
and $\Sigma^x_{k+1} = A \Sigma^x_k A^\top + \Sigma^w,~ \Sigma^x_0 = 0$, denotes the state covariance matrix.
We define the error system
\begin{equation}\label{eq:error_system}
    e_{k+1} = x_{k+1} - z_{k+1} = Ae_k + w_k,~ e_0 = 0,
\end{equation}
with $\E[e_{k+1}]=0$ and $\E\big[e_{k+1}e_{k+1}^\top\big]=\Sigma^x_{k+1}$ such that $x_{k+1} = z_{k+1} + e_{k+1}$.
Note that if $\mathcal{P}^w$ is known exactly, the state distributions $\mathcal{P}_k^x$ can be computed exactly.

The control objective is to (optimally) stabilize system \eqref{eq:system} at the origin while satisfying state and input constraints given via the polytopes
\begin{subequations}
\begin{align}
    \mathcal{X} &= \{ x \in \R{\n{x}} \mid H^x x \leq h^x \}, ~0 \in \mathrm{int}(\mathcal{X}), ~\text{and} \label{eq:state_constraints} \\ 
    \mathcal{U} &= \{ u \in \R{\n{u}} \mid H^u u \leq h^u \}, ~0 \in \mathrm{int}(\mathcal{U}). \label{eq:input_constraints}
\end{align}
\end{subequations}
Herein, $H^x \in \R{\n{c}^x \times \n{x}}$ and $h^x \in \R{\n{c}^x}$ define $\n{c}^x \in \mathbb{N}$ half-space constraints on the states, and $H^u \in \R{\n{c}^u \times \n{u}}$ and $h^u \in \R{\n{c}^u}$ define $\n{c}^u \in \mathbb{N}$ half-space constraints on the inputs.
% Clearly, $\mathcal{X}$ and $\mathcal{U}$ are compact and convex sets, for which we additionally require in the following that $0 \in \mathrm{int}($
Since Assumption \ref{as:disturbance} includes distributions $\mathcal{P}^w$ with unbounded support, we cannot expect to achieve deterministic satisfaction of state constraints at all times.
Rather, we aim to design a controller that achieves $\bbP[x_k \in \mathcal{X}] \geq 1 - \delta$ for a maximum tolerable probability of constraint violation $\delta \in (0, 1)$.
Moreover, given the additive and stochastic nature of $w_k$, it is intuitively clear that it is not possible to stabilize system \eqref{eq:system} exactly at the origin and in a deterministic manner.
Rather, a suitable controller needs to achieve probabilistic stabilization of system \eqref{eq:system} in a preferably small neighborhood of the origin.
Given these requirements, we employ stochastic model predictive control.

\subsection{Stochastic Model Predictive Control Formulation}

Stochastic MPC relies on the repeated solution of an underlying stochastic optimal control problem of the form
\begin{subequations}\label{eq:ocp}
\begin{align}
    \min_{\mathbf{{u}}_k} \Bigg\{ \! J_N&(\mathbf{{u}}_k, \mathbf{{x}}_k) \! = \! \E \! \Bigg[ \sum_{i=0}^{N-1} \! \ell({x}_{i \mid k}, {u}_{i \mid k}) \! + \! V_f({x}_{N \mid k}) \! \Bigg] \! \Bigg\} \label{eq:ocp_cost} \\
    \text{s.~t.} \qquad & \forall i \in \mathcal{I}_{0:N-1}: \nonumber \\
                        & ~~ {x}_{i+1 \mid k} = A {x}_{i \mid k} + B {u}_{i \mid k} + {w}_{i \mid k}, {x}_{0 \mid k} = x_k \label{eq:ocp_dynamics} \\
                        & ~~ {w}_{i \mid k} \sim \mathcal{P}^w, \mathcal{P}^w \in \mathcal{A}^w \label{eq:ocp_disturbance} \\
                        & ~~ {u}_{i \mid k} \in \mathcal{U} \label{eq:ocp_input_constraints} \\
                        & ~~ \bbP[{x}_{i \mid k} \in \mathcal{X}] \geq 1 - \delta \label{eq:ocp_chance_constraints} \\
                        & \bbP[{x}_{N \mid k} \in \mathcal{X}_f] \geq 1 - \delta. \label{eq:ocp_chance_constraint_terminal}
\end{align}
\end{subequations}
Therein, $N \in \mathbb{N}, N < \infty$, is the (prediction and control) horizon and subscript ${i \mid k}$ denotes a prediction $i$ steps ahead of $k$.
The cost function $J_N(\mathbf{{u}}_k, \mathbf{{x}}_k)$ in \eqref{eq:ocp_cost}, consisting of the stage cost $\ell: \R{\n{x}} \times \R{\n{u}} \to \R{}$ and the terminal cost $V_f: \R{\n{x}} \to \R{}$, is minimized with respect to the input sequence $\mathbf{{u}}_k = \begin{bmatrix} {u}_{0 \mid k}, \ldots, {u}_{N-1 \mid k} \end{bmatrix}$.
The state sequence $\mathbf{{x}}_k = \begin{bmatrix} {x}_{0 \mid k}, \ldots, {x}_{N \mid k} \end{bmatrix}$ is determined via the dynamics constraint \eqref{eq:ocp_dynamics}, initialized at the current state $x_k$, whereupon \eqref{eq:ocp_disturbance} accounts for the nature of the disturbances.
% Note that $\hat{x}_{i \mid k}$ are future states that depend on future disturbances $\hat{w}_{i \mid k}$, which have not realized yet.
% Thus, the state prediction ...
The state predictions ${x}_{i \mid k} \sim \mathcal{P}^x_i$ are random vectors, which is accounted for by the expectation operator in \eqref{eq:ocp_cost} and the chance constraint formulations \eqref{eq:ocp_chance_constraints} and \eqref{eq:ocp_chance_constraint_terminal}, where $\mathcal{X}_f \subseteq \mathcal{X}, 0 \in \mathrm{int}(\mathcal{X}_f),$ is the terminal region. % with terminal region $\mathcal{X}_f \subseteq \mathcal{X}, 0 \in \mathrm{int}(\mathcal{X}_f)$.
% In contrast to state constraints, 
The input constraints \eqref{eq:ocp_input_constraints} can be enforced deterministically in the considered set-up.
A similar probabilistic treatment is necessary when employing an ancillary, pre-stabilizing controller \cite{Hewing2020b}.

Given the infinite-dimensional chance constraints \eqref{eq:ocp_chance_constraints} and \eqref{eq:ocp_chance_constraint_terminal}, the optimal control problem \eqref{eq:ocp} is computationally intractable.
To remedy from this situation, the chance constraints are reformulated as tightened constraints on the nominal dynamics \eqref{eq:nominal_system}, exploiting probabilistic reachable sets of the error dynamics \eqref{eq:error_system}.
\begin{definition}\label{def:i_step_PRS}(From \cite{Hewing2020a}).
    A set $\mathcal{R}_i^{1-\delta}$ is said to be an $i$-step probabilistic reachable set of probability level $1-\delta$ for the error dynamics \eqref{eq:error_system} if
    \begin{equation}
        \bbP[e_i \in \mathcal{R}_i^{1-\delta} \mid e_0 = 0] \geq 1 - \delta.
    \end{equation}
\end{definition}
\noindent Particularly, $\bbP[\hat{x}_{i \mid k} \in \mathcal{X}] \geq 1- \delta$ is equivalent to $z_{i \mid k} \in \mathcal{Z}_i$ for $\mathcal{Z}_i := \mathcal{X} \ominus \mathcal{R}_i^{1 - \delta}$ \cite{Hewing2020a}.\footnote{Since system \eqref{eq:system} is time-invariant and $w_k$ is i.i.d. over time, the sets $\mathcal{Z}_{i \mid k} = \mathcal{Z}_i, \mathcal{R}^{1-\delta}_{i \mid k} = \mathcal{R}_i^{1-\delta}$ are independent of $k$ and $\hat{x}_{0 \mid k} = x_k$.}
Given $\mathcal{X}$ according to \eqref{eq:state_constraints} and exploiting Boole's inequality, we find that \cite{Hewing2020a}
\begin{align}\label{eq:tightened_constraints}
    \mathcal{Z}_i = \left\{ z \mid H^{x}_{j:} z \leq h^{x}_j - \psi_j \sqrt{{H^{x}_{j:}} \Sigma^{x}_i {H^{x}_{j:}}^\top}, j \in \mathcal{I}_{1 : \n{c}^{x}} \right\}
\end{align}
for $i \in \mathcal{I}_{0:N-1}$, where $H^x_{j:}$ and $h^x_j$ denote the $j^\text{th}$ row of $H^x$ and the $j^\text{th}$ entry of $h^x$, respectively. 
The tightening factors $\psi_j$ are given by $\psi_j = \Phi_x^{-1}\big(1 - \frac{\delta}{\n{c}^x}\big)$ if the state distribution $\mathcal{P}^x$ and its (standardized) inverse cumulative density function $\Phi^{-1}$ are exactly known \cite{Ono2008}, or by $\psi_j = \sqrt{\frac{\n{c}^x-\delta}{\delta}}$ if only $\mathcal{P}^x_i \in \mathcal{A}^x_i$ is known \cite{Calafiore2006}.
The computationally tractable approximation of \eqref{eq:ocp} then reads
\begin{subequations}\label{eq:smpc}
\begin{align}
    \min_{\mathbf{{u}}_k} \Bigg\{ \! J_N&(\mathbf{{u}}_k, \mathbf{{x}}_k) \! = \! \E \! \Bigg[ \sum_{i=0}^{N-1} \! \ell({x}_{i \mid k}, {u}_{i \mid k}) \! + \! V_f({x}_{N \mid k}) \! \Bigg] \! \Bigg\} \label{eq:smpc_cost} \\
    \text{s.~t.} \qquad & \forall i \in \mathcal{I}_{0:N-1}: \nonumber \\
                        & ~~ z_{i+1 \mid k} = A z_{i \mid k} + B {u}_{i \mid k}, z_{0 \mid k} = x_k \label{eq:smpc_dynamics} \\
                        % & ~~ \Sigma^x_{i+1} = A \Sigma^x_i A^\top + \Sigma^w, \Sigma_0^x = 0 \label{eq:smpc_covariance} \\
                        % & ~~ \hat{u}_{i \mid k} \in \mathcal{U} \label{eq:smpc_input_constraints} \\
                        % & ~~ z_{i \mid k} \in \mathcal{Z}_i \label{eq:smpc_state_constraints} \\
                        & ~~ {u}_{i \mid k} \in \mathcal{U}, z_{i \mid k} \in \mathcal{Z}_i \label{eq:smpc_constraints} \\
                        & z_{N \mid k} \in \mathcal{Z}_f. \label{eq:smpc_terminal_constraint}
\end{align}
\end{subequations}
Herein, $\mathcal{Z}_f \subseteq \mathcal{X}_f \ominus \mathcal{R}_N^{1-\delta}$ is a suitably chosen terminal region.
Problem \eqref{eq:smpc} exhibits a solution, i.e., is feasible for $x_k$, if and only if there exists at least one input sequence such that all constraints are satisfied. 
We call such an input sequence admissible according to the following definition.

\begin{definition}(From \cite{Grune2017}).
    An input sequence $\mathbf{u}$ is called admissible for optimal control problem \eqref{eq:smpc}, state $x \in \mathcal{X}$ and horizon $N$, if $\forall i \in \mathcal{I}_{0:N-1}: u_i \in \mathcal{U}, z_{i} \in \mathcal{Z}_{i}$ and $z_N \in \mathcal{Z}_f$.
    The set of admissible control sequences for state $x \in \mathcal{X}$ and horizon $N$ is denoted by $\mathbb{U}^N(x)$.
\end{definition}

\noindent We define the set of states for which problem \eqref{eq:smpc} is feasible as follows.

\begin{definition}(From \cite{Grune2017}).
    The set of feasible initial conditions for MPC \eqref{eq:smpc} is given by $\mathcal{X}_0 := \{ x \in \mathcal{X} \mid \mathbb{U}^N(x) \neq \emptyset \}$.
\end{definition}

\noindent Given that $x_k \in \mathcal{X}_0$, we denote the optimal input sequence obtained from solving \eqref{eq:smpc} by ${\mathbf{u}}_k^*$.
Therefrom, the first element ${u}_{0 \mid k}$ is applied to the system, which defines the MPC policy
\begin{equation}\label{eq:mpc_law}
    \mu_\mathrm{MPC}(x_k) = {u}_{0 \mid k}.
\end{equation}
The corresponding optimal value of \eqref{eq:smpc} is defined as follows.

\begin{definition}\label{def:optimal_value}
    The optimal value function $V_N(x_k)$ of \eqref{eq:smpc} is given by
    \begin{align}\label{eq:VN_MPC}
        V_N(x_k) = \begin{cases} J_N(\mathbf{u}_k^*, x_k) & \text{if~} x_k \in \mathcal{X}_0 \\
        \infty & \text{if~} x_k \notin \mathcal{X}_0 \end{cases}.
    \end{align}
\end{definition}

\subsection{Remarks on Feasibility and Stability of SMPC}\label{sec:feasibility_stability_remarks}

Since we consider unbounded, stochastic disturbances $w_k$, it is impossible to guarantee that, if $x_k \in \mathcal{X}_0$, applying the MPC law \eqref{eq:mpc_law} will lead to $x_{k+1} \in \mathcal{X}_0$.
Moreover, existence of an admissible input sequence at time point $k$ does not imply that its execution will remain strictly safe along the horizon, making it impossible to exploit the optimal solution at time point $k$ to construct an admissible candidate input sequence for time point $k+1$ \cite{Kouvaritakis2010}.
Hence, strict guarantees on repeated feasibility of the stochastic MPC \eqref{eq:smpc} are impossible to obtain.
In consequence, a back-up controller $u_k = \mu_{\mathrm{BackUp}}(x_k)$ is required for computing suitable control actions if $x_k \notin \mathcal{X}_0$.
This results in the overall control strategy
\begin{align}\label{eq:control_policy_general}
    u_k = \begin{cases}
              \mu_\mathrm{MPC}(x_k)    & \text{if}~ x_k \in \mathcal{X}_0 \\
              \mu_\mathrm{BackUp}(x_k) & \text{if}~ x_k \notin \mathcal{X}_0
          \end{cases}.
\end{align}

Standard approaches for proving (robust) stability of MPC schemes rely on the explicit construction of an admissible input sequence $\mathbf{{u}}_{k+1}^+$ for the next time point $k+1$ based on the current optimal solution $\mathbf{{u}}_k^*$, see, e.g., \cite{Rawlings2019}.
Since constructing such a sequence that is guaranteed to be admissible for the next horizon is impossible (as discussed above), standard stability results do not apply.
To provide ISSp results for system \eqref{eq:system} under MPC \eqref{eq:smpc} in the following, we adopt the strategy proposed in \cite{Primbs2009}, circumventing the requirement of repeated feasibility of MPC \eqref{eq:smpc}.
Based thereon and for a suitable choice of the back-up controller, we will show ISSp of system \eqref{eq:system} under control law \eqref{eq:control_policy_general}.
\section{Probabilistically Input-to-State Stable SMPC}\label{sec:issp}

We first introduce the concept of ISSp.
% Thereafter, ISSp of the autonomous system is considered, followed by investigating ISSp of the system under the MPC.
Thereafter, ISSp of the system under the back-up controller is considered, followed by investigating ISSp under the MPC.
We finish by concluding ISSp of the system under the proposed control policy.

\subsection{Input-to-State Stability in Probability}

Preliminarily to introducing ISSp, we consider the characterization of the random disturbance $w_k$ via $L^p$ spaces.

\begin{definition}\label{def:Lp_norm}(From \cite{Culbertson2023}).
    A random vector $\chi \sim \mathcal{P}^\chi$ belongs to $L^p$ for some $p > 0$, denoted by $\chi \in L^p$, if it holds that
    \begin{equation}
        \| \chi \|_{L^p} := \mathbb{E}[ \| \chi \|^p]^\frac{1}{p} < \infty.
    \end{equation}
\end{definition}

Note that by Assumption \ref{as:disturbance}, $w_k \in L^2$ since $\Sigma^w \preceq_L \nu \mathbb{I}, 0 < \nu < \infty$ ensures $\Tr[\Sigma^w] < \infty$ and $\| w_k \|_{L^2} = \sqrt{\Tr[\Sigma^w]}$.
This further implies $w_k \in L^p$ for $p > 2$ \cite{Culbertson2023}.
Since $w_k$ is i.i.d. over time, and hence its $L^p$-norm is constant over time, we denote $\| w_k \|_{L^p} = \| w \|_{L^p}, \forall k,$ for clarity.
Given the disturbance $w_k \in L^p$, ISSp is defined as follows.

\begin{definition}\label{def:issp}(From \cite{Culbertson2023}).
    A system $x_{k+1} = f(x_k, w_k)$ is input-to-state stable in probability with respect to $L^p$, if for any $\varepsilon \in (0, 1)$, $M \in \mathbb{N}$ and $w_k \in L^p$, there exist functions $\beta \in \mathcal{KL}$ and $\rho \in \mathcal{K}$ such that
    \begin{align}
        \bbP[ \|x_{k+i} \| \leq \beta( \|x_k\|, i) + \rho(\| w \|_{L^p}), \forall i \! \leq \! M] \! \geq \! 1 \! - \! \varepsilon.
    \end{align}
\end{definition}

Loosely speaking, if a system $x_{k+1} = f(x_k, w_k)$ is ISSp, it converges for finite horizon $M$ and with probability arbitrarily close to 1 to a neighborhood of the equilibrium of the undisturbed dynamics.
The size of this neighborhoor is determined by the $L^p$-norm of the disturbance $w_k$.
ISSp can be characterized using ISSp Lyapunov functions.

\begin{definition}\label{def:issp_lyapunov}(From \cite{Culbertson2023}).
    A continuous function $V: \R{\n{x}} \to [0, \infty)$ is an ISSp Lyapunov function for $x_{k+1} = f(x_k, w_k)$ on $\Omega \subseteq \R{\n{x}}$ if there exist functions $\alpha_1, \alpha_2, \kappa \in \mathcal{K}_\infty$ and $\varphi \in \mathcal{K}$ such that
    \begin{align}
        &\alpha_1(\| x_k \|) \leq V(x_k) \leq \alpha_2(\| x_k \|), \label{eq:ISSp_lyapunov_function_bounds} \\
        &\E[V(x_{k+1}) \! - \! V(x_k) \mid x_k] \leq - \kappa(V(x_k)) + \varphi(\| w \|_{L^p})
    \end{align}
    hold for all $x_k \in \Omega$ and $w_k \in L^p$.
\end{definition}

\begin{proposition}\label{prop:issp_lyapunov}(From \cite{Culbertson2023}).
    If there exists an ISSp Lyapunov function for $x_{k+1} = f(x_k, w_k)$ on $\Omega$, then $x_{k+1} = f(x_k, w_k)$ is ISSp on $\Omega$.    
\end{proposition}

\subsection{ISSp under the Back-Up Controller}\label{sec:issp_autonomous}

We require that system \eqref{eq:system} is ISSp under the back-up controller, and appropriate back-up controller design is in general a challenging task under input constraints, see \cite{Brudigam2022, Paulson2017}.
In the following, we consider the uncontrolled version of system \eqref{eq:system} and its nominal counterpart, given by
\begin{align}
    x_{k+1} & = A x_k + w_k,~ \text{and} \label{eq:disturbed_autonomous_dynamics} \\
    z_{k+1} & = A z_k, \label{eq:nominal_autonomous_dynamics}
\end{align}
and restrict the analysis to asymptotically stable systems, for which $\mu_\mathrm{BackUp}(x_k) = 0$ is a valid choice.

\begin{assumption}\label{as:nominal_autonomous_stability}
    The autonomous nominal dynamics \eqref{eq:nominal_autonomous_dynamics} are asymptotically stable and there exists $P \in \R{\n{x} \times \n{x}}, P \succ 0, P^\top = P$ such that ${V}(z_k) = z_k^\top P z_k$ is a Lyapunov function for \eqref{eq:nominal_autonomous_dynamics}. 
    Thus, ${V}(z_{k+1}) - {V}(z_{k}) = z_{k}^\top A^\top P A z_k - z_k P z_k \leq 0$ and $A^\top P A - P \prec 0$.
\end{assumption}
\noindent We obtain the following stability assertion for system \eqref{eq:disturbed_autonomous_dynamics}.
\begin{theorem}\label{th:issp_autonomous_system}
    Under Assumptions \ref{as:disturbance} and \ref{as:nominal_autonomous_stability}, the autonomous disturbed system \eqref{eq:disturbed_autonomous_dynamics} is ISSp with respect to $L^2$ and ${V}(x_k) = x_k^\top P x_k$ is an ISSp Lyapunov function for \eqref{eq:disturbed_autonomous_dynamics}.
\end{theorem}

\begin{proof}
    Given $x_k$, we find that ${V}(x_k) = x_k^\top P x_k$ and ${V}(x_{k+1}) = (A x_k + w_k)^\top P (A x_k + w_k)$. 
    Thus,
    \begin{align*}
        \E [{V}(x_{k+1}) & - {V}(x_k) \mid x_k] \\
                                & = x_k^\top (A^\top P A - P) x_k + \E[w_k^\top P w_k] \\
                                & = - x_k^\top (P - A^\top P A) x_k + \Tr(P \Sigma^w) \\  
                                & = - \| (P - A^\top P A)^{\frac{1}{2}} x_k \|_2^2 + \Tr(P \Sigma^w) \\
                                & \leq - \sigma_\mathrm{min}^2 \big((P - A^\top P A)^{\frac{1}{2}} \big) \| x_k \|_2^2 + \Tr(P \Sigma^w).
    \end{align*} 
    Herein, $\sigma_\mathrm{min}(\cdot)$ denotes the minimum singular value.
    Since for ${V}(x_k) = x_k^\top P x_k = \| P^\frac{1}{2} x_k \|_2^2$ it holds that ${V}(x_k) \leq \sigma_\mathrm{max}^2\big(P^\frac{1}{2}\big) \| x_k \|_2^2$, we find that
    \begin{align*}
        \E [{V}&(x_{k+1}) - {V}(x_k) \mid x_k] \\
                       & \leq - \underbrace{\frac{\sigma_\mathrm{min}^2 \big((P \! - \! A^\top P A)^{\frac{1}{2}} \big)}{\sigma_\mathrm{max}^2 \big( P^\frac{1}{2} \big)} {V}(x_k)}_{=: \kappa( {V}(x_k) )} + \underbrace{\Tr(P \Sigma^w)}_{=: \varrho(\| w_k \|_{L^2})}. 
    \end{align*}
    Clearly, $\kappa \in \mathcal{K}_\infty$ and $\varrho \in \mathcal{K}$ and hence, $V$ is an ISSp Lyapunov function for \eqref{eq:disturbed_autonomous_dynamics}.
    Then, according to Proposition \ref{prop:issp_lyapunov}, \eqref{eq:disturbed_autonomous_dynamics} is ISSp.
\end{proof}

% Note that Assumption \ref{as:nominal_autonomous_stability} includes the case that the nonautonomous system is pre-stabilized by a linear controller $u_k = v_k - K x_k, K \in \R{\n{u} \times \n{x}}$, such that $A = \Tilde{A} - BK$ describes the closed-loop dynamics of a system with original dynamics matrix $\Tilde{A}$.
% However, constraints on $u_k$ cannot be satisfied at all times in such case since the pre-stabilizing controller is an unconstrained one.

\subsection{ISSp under the Stochastic MPC}\label{sec:issp_mpc}

In order to prove stability of the stochastic MPC formulation, we require the following assumptions about the stochastic MPC \eqref{eq:smpc} to hold.
\begin{assumption}\label{as:quadratic_cost}
    The stage cost $\ell(x,u)$ and the terminal cost $V_f(x)$ are quadratic functions given by
    \begin{align}
        \ell(x,u) & = x^\top Q x + u^\top R u,~\text{and} \label{eq:stage_cost} \\
        V_f(x) & = x^\top Q_f x \label{eq:terminal_cost}
    \end{align}
    with $Q, Q_f \in \R{\n{x} \times \n{x}}$, $Q, Q_f \succ 0$, $Q = Q^\top$, $Q_f = Q_f^\top$, and $R \in \R{\n{u} \times \n{u}}$, $R \succ 0$, $R = R^\top$.
    Moreover, there exist $\alpha_1, \alpha_2, \alpha_3 \in \mathcal{K}_\infty$ such that
    \begin{align}
        \ell(x,u) & \geq \alpha_3( \| x \| ),~\text{and} \\
        \alpha_1( \| x \|) & \leq V_f(x) \leq \alpha_2( \| x \|).
    \end{align}
\end{assumption}

\begin{proposition}\label{prop:optimal_value_function}(Theorem 4 from \cite{Sakizlis2007}). 
Under Assumption \ref{as:quadratic_cost}, the optimal value function \eqref{eq:VN_MPC} of MPC \eqref{eq:smpc} is continuous, convex, and piecewise quadratic on $\mathcal{X}_0$. Further, the MPC law \eqref{eq:mpc_law} is continuous and piecewise affine, and the set $\mathcal{X}_0$ is convex.
\end{proposition}

Assumption \ref{as:quadratic_cost} enables to explicitly state the cost function $J_N(\mathbf{u}_k, \mathbf{x}_k)$ in terms of a nominal part and an uncertainty-related part as
\begin{align}
    J_N(\mathbf{u}_k, \mathbf{x}_k) & = \sum_{i=0}^{N-1} \| z_{i \mid k} \|_Q^2 + \| u_{i \mid k} \|_R^2 + \Tr(Q \Sigma_i^x) \nonumber \\
                                    & \qquad \qquad \qquad + \| z_{N \mid k} \|_{Q_f}^2 + \Tr(Q_f \Sigma_N^x) \\
                                    & = \hat{J}_N(\mathbf{u}_k, \mathbf{z}_k) + c, 
\end{align}
where the uncertainty-related part $c = \sum_{i=0}^{N-1} \Tr(Q \Sigma_i^x) + \Tr(Q_f \Sigma_N^x)$ is constant over time since system \eqref{eq:system} is time-invariant and $w_k$ is i.i.d. over time.
Thus, the optimal value function \eqref{eq:VN_MPC} can be decomposed over $\mathcal{X}_0$ into
\begin{align}\label{eq:VN_decomposition}
    V_N(x_k) = \hat{V}_N(x_k) + c,
\end{align}
where $\hat{V}_N(x_k) = \hat{J}_N(\mathbf{u}_k^*, \mathbf{z}_k^*)$ and $c$ is as before.
In consequence of \eqref{eq:VN_decomposition} and Proposition \ref{prop:optimal_value_function}, the nominal optimal value function $\hat{V}_N(x_k)$ is continuous, convex and piecewise quadratic.

\begin{assumption}\label{as:standard_stability}
    There exists a terminal controller $u = -K_f z$ with gain matrix $K_f \in \R{\n{u} \times \n{x}}$ such that
    \begin{enumerate}
        \item[(i)] $\forall z \in \mathcal{Z}_f: u = -K_f z \in \mathcal{U}$,
        \item[(ii)] the terminal region $\mathcal{Z}_f$ is forward invariant for the nominal system \eqref{eq:nominal_system} under the terminal controller, i.e., $\forall z \in \mathcal{Z}_f: (A-BK_f)z \in \mathcal{Z}_f$, and
        \item[(iii)] the terminal cost \eqref{eq:terminal_cost} is a local Lyapunov function on $\mathcal{Z}_f$ for the nominal system \eqref{eq:nominal_system} under the terminal controller, i.e., $\forall z \in \mathcal{Z}_f$:
        \begin{equation}
            V_f((A-BK_f)z) - V_f(z) \leq - \ell(z, -K_f z).
        \end{equation}
    \end{enumerate}
\end{assumption}

The quadratic stage and terminal cost functions, as specified by Assumption \ref{as:quadratic_cost}, are standard choices in MPC for a wide variety of regulation (and in extension also tracking) problems.
% Thus, Assumption \ref{as:quadratic_cost} is not restrictive.
Further, Assumption \ref{as:standard_stability} is a standard assumption in MPC, see, e.g., \cite{Rawlings2019, Mayne2019}.
Note that Assumption \ref{as:standard_stability} is not restrictive since it is formulated with respect to the nominal dynamics, for which methods to compute a terminal set, cost function and controller with the required properties are available, see, e.g., \cite{Maiworm2015}.

For proving ISSp of system \eqref{eq:system} when applying MPC \eqref{eq:smpc} on $\mathcal{X}_0$, we make use of the Lyapunov function candidate
\begin{align}\label{eq:lyapunov_candidate}
    \Tilde{V}_N(x_k) = \begin{cases} \hat{V}_N(x_k) & \text{if~} x_k \! \in \! \mathcal{X}_0 \\ \max_{a \in [0,1]} \{ \hat{V}_N(ax_k) \mid ax_k \! \in \! \mathcal{X}_0 \} & \text{if~} x_k \! \notin \! \mathcal{X}_0 \end{cases}.
\end{align}
Note that \eqref{eq:lyapunov_candidate} coincides with the nominal optimal value function $\hat{V}_N$ if $x_k \in \mathcal{X}_0$. 
We exploit $\hat{V}_N$ instead of $V_N$ directly to achieve compliance of \eqref{eq:lyapunov_candidate} with \eqref{eq:ISSp_lyapunov_function_bounds} in the following.
Since we have to account for the possibility that $x_{k+1}$ becomes infeasible, $\Tilde{V}_N$ is a continuous extension of $\hat{V}_N$.
Particularly, if $x_{k+1} \notin \mathcal{X}_0$, the value of $\Tilde{V}_N(x_{k+1})$ is given by the value $\hat{V}_N(\partial x_{k+1})$, where $\partial x_{k+1} \in \partial \mathcal{X}_0$ is the point on the boundary of $\mathcal{X}_0$ that simultaneously lies on the straight line connecting the origin and $x_{k+1}$.
This point, $\partial x_{k+1} \in \partial \mathcal{X}_0$, is unique since $\mathcal{X}_0$ and $\hat{V}_N$ (on $\mathcal{X}_0$) are convex by Proposition \ref{prop:optimal_value_function}.
Furthermore, by Proposition \ref{prop:optimal_value_function}, $\hat{V}_N$ is continuous on $\mathcal{X}_0$, including $\partial \mathcal{X}_0$, which implies continuity of $\Tilde{V}_N$.
Using \eqref{eq:lyapunov_candidate}, we can state the following stability result.

\begin{theorem}\label{th:issp_controlled_system}
    Let $x_k \in \mathcal{X}_0$ be feasible for the MPC defined by \eqref{eq:smpc}. Then, under Assumptions \ref{as:disturbance}, \ref{as:quadratic_cost} and \ref{as:standard_stability}, $\Tilde{V}_N$ as defined in \eqref{eq:lyapunov_candidate} is an ISSp Lyapunov function for system \eqref{eq:system} under MPC law \eqref{eq:mpc_law} on $\mathcal{X}_0$.
\end{theorem}

\begin{proof}
    The first part of the proof follows the arguments used in \cite{Primbs2009}.
    Since $x_k \in \mathcal{X}_0$, MPC problem \eqref{eq:smpc} is feasible for $x_k$.
    Solving MPC \eqref{eq:smpc} for $x_k \in \mathcal{X}_0$ yields the optimal input sequence $\mathbf{\hat{u}}_k^*$, the corresponding nominal state sequence $\mathbf{z}_k^*$ and the nominal optimal value
    \begin{align}\label{eq:theorem2_VN_k}
        \hat{V}_N(x_k) & = \sum_{i=0}^{N-1} \ell(z_{i \mid k}^*, u_{i \mid k}^*) + V_f(z^*_{N \mid k}) \nonumber \\
                       & = \sum_{i=0}^{N-1} \| z_{i \mid k}^* \|_Q^2 + \| u_{i \mid k}^* \|_R^2 + \| z_{N \mid k}^* \|_{Q_f}^2.
    \end{align}
    Define $\mathbb{F}:= \{ z_{1 \mid k} \mid z_{0 \mid k} = x_k, \mathbf{u}_{1:N-1 \mid k} \in \mathbb{U}^{N-1}(z_{1 \mid k}) \}$, which is the set of all states $x_{k+1}$ for which the remaining input sequence $\mathbf{u}_{1:N-1 \mid k}^*$ is feasible at time point $k+1$ for horizon $N-1$ when $u^*_{0 \mid k}$ is applied at time point $k$.
    It follows from Assumption \ref{as:standard_stability}, that the input sequence $\mathbf{u}_{k+1}^+$ with
    \begin{align}\label{eq:theorem2_u_plus}
        u^+_{i \mid k+1} = \begin{cases} u_{i+1 \mid k}^* & \text{for~} i \in \mathcal{I}_{0:N-2} \\ -K_f z_{N-1 \mid k+1}^+ & \text{for~} i = N-1 \end{cases}
    \end{align}
    is feasible for $x_{k+1} \in \mathbb{F}$.\footnote{This refers to feasibility for MPC \eqref{eq:smpc} at time point $k+1$ and does not imply feasibility at later time steps when successively executing the input sequence, as discussed in Section \ref{sec:feasibility_stability_remarks}.}
    Thus, and since $u_{k+1}^+$ is not necessarily optimal, we find on the set $\mathbb{F}$ that
    \begin{align}
        \E[ \1{\bbF}\{\hat{V}_N(x_{k+1})\} ] & \leq \E[ \1{\bbF}\{ \hat{J}_N(x_{k+1}, \mathbf{u}_{k+1}^+) \} ] \nonumber \\
                                       & = \E[\hat{J}_N(x_{k+1}, \mathbf{u}_{k+1}^+)] \bbP[x_{k+1} \in \bbF] \nonumber \\
                                       & \leq \E[\hat{J}_N(x_{k+1}, \mathbf{u}_{k+1}^+)]
    \end{align}
    Analogously to \eqref{eq:theorem2_VN_k}, it holds that
    \begin{multline}\label{eq:theorem2_JN_kp1}
        \hat{J}_N(x_{k+1}, \mathbf{u}_{k+1}^+) = \sum_{i=0}^{N-1} \| z_{i \mid k+1}^+ \|_Q^2 + \| u_{i \mid k+1}^+ \|_R^2 \\ + \| z_{N \mid k+1}^+ \|_{Q_f}^2.
    \end{multline}
    Furthermore, exploiting that $z_{0 \mid k+1}^+ = x_{k+1} = z_{0 \mid k}^* + w_k$, we find by iterating the dynamics \eqref{eq:nominal_system} initialized at $x_{k+1}$ that
    \begin{align}\label{eq:theorem2_dynamics}
        z_{i \mid k+1}^+ = z_{i+1 \mid k}^* + A^i w_k.
    \end{align}
    Substituting \eqref{eq:theorem2_u_plus} and \eqref{eq:theorem2_dynamics} in  \eqref{eq:theorem2_JN_kp1} and using $\|z_{i \mid k+1}^* \|_Q \leq \| z_{i+1 \mid k}^* \|_Q + \| A^i w_k \|_Q$, where the latter is obtained from applying the triangle inequality, leads to
    \begin{align}
        & \E[ \1{\mathbb{F}} \{ \hat{V}_N(x_{k+1}) \} ] - \hat{V}_N(x_k) \nonumber \\ 
        & \leq \E \Bigg[- \ell(x_k, u_{0 \mid k}^*) +  \sum_{i=0}^{N-1} \|A^i w_k \|_Q^2 \nonumber \\ 
        & ~~~~+ \ell(z_{N-1 \mid k+1}^+, u_{N-1 \mid k+1}^+) + V_f(z_{N+1 \mid k}) - V_f(z_{N \mid k}) \Bigg]. \nonumber
    \end{align}
    As for $x_{k+1} \in \bbF$, it holds that $z_{N-1 \mid k+1} \in \mathcal{Z}_f$, we have by Assumption \ref{as:standard_stability} that $\ell(z_{N-1 \mid k+1}^+, u_{N-1 \mid k+1}^+) + V_f(z_{N+1 \mid k}) - V_f(z_{N \mid k}) \leq 0$.
    Thus,
    \begin{align}\label{eq:theorem2_bound_F_1}
        \E[ \1{\mathbb{F}} \{ \hat{V}_N(&x_{k+1}) \} ] - \hat{V}_N(x_k) \nonumber \\
        & \leq \E \Bigg[- \ell(x_k, u_{0 \mid k}^*) +  \sum_{i=0}^{N-1} \|A^i w_k \|_Q^2 \Bigg] \nonumber \\
        & = - \ell(x_k, u_{0 \mid k}^*) + \sum_{i=0}^{N-1} \Tr({A^i}^\top Q A^i \Sigma^w).
    \end{align}
    In consequence of Proposition \ref{prop:optimal_value_function}, and since $\hat{V}_N$ clearly satisfies \eqref{eq:ISSp_lyapunov_function_bounds}, there exists $\Lambda \in \R{\n{x} \times \n{x}}$, $\Lambda \succeq 0$ such that $V_N(x_k) \leq x_k^\top \Lambda x_k \leq \sigma_\mathrm{max}^2(\Lambda^\frac{1}{2}) \|x_k\|_2^2$.
    This implies that $\|x_k\| \geq \Tilde{\alpha}(V_N(x_k))$ for $\Tilde{\alpha} \in \mathcal{K}_\infty$.
    Since by Assumption \ref{as:quadratic_cost}, $\ell(x_k, u_{0 \mid k}^*) \geq \alpha_3( \| x_k \| )$, there exists $\kappa \in \mathcal{K}_\infty$ such that $\ell(x_k, u_{0 \mid k}^*) \geq \kappa(V_N(x_k))$.
    Exploiting further that for $x_k \in \mathcal{X}_0$ and $x_{k+1} \in \bbF \subset \mathcal{X}_0$, it holds that $\Tilde{V}_N(x_k) = \hat{V}_N(x_k)$ and $\Tilde{V}_N(x_{k+1}) = \hat{V}_N(x_{k+1})$, we find that
    \begin{multline}\label{eq:theorem2_bound_F_2}
        \E[ \1{\mathbb{F}} \{ \tilde{V}_N(x_{k+1}) \} ] - \tilde{V}_N(x_k) \\
        = - \kappa(\Tilde{V}_N(x_{k})) + \sum_{i=0}^{N-1} \Tr({A^i}^\top Q A^i \Sigma^w).
    \end{multline}
    Using the identity $\E[\Tilde{V}_N(x_{k+1})] = \E[\1{\bbF}\{ \Tilde{V}_N(x_{k+1}) \}] + \E[\1{\bbF^c}\{ \tilde{V}_N(x_{k+1}) \}]$, where $\bbF^c$ denotes the complement of $\bbF$ on $\R{\n{x}}$, and $\E[\1{\bbF^c}\{ \tilde{V}_N(x_{k+1}) \}] = \E[\tilde{V}_N(x_{k+1})] \bbP[x_{k+1} \notin \bbF]$ \eqref{eq:theorem2_bound_F_2} yields
    \begin{multline}\label{eq:theorem2_stability_1}
        \E[ \Tilde{V}_N(x_{k+1}) - \Tilde{V}_N(x_k) \mid x_k ] \leq - \kappa(\Tilde{V}_N(x_k)) \\ + \sum_{i=0}^{N-1} \Tr({A^i}^\top Q A^i \Sigma^w) + \lambda \bbP[x_{k+1} \notin \bbF],
    \end{multline}
    where $\lambda = \max_{x \in \mathcal{X}_0} \{ \hat{V}_N(x) \}$ is an upper bound of $\Tilde{V}_N(x)$ according to \eqref{eq:lyapunov_candidate}.
    It remains to bound $\bbP[x_{k+1} \notin \bbF]$.
    To this end, note that $z_{1 \mid k}^* \in \mathrm{int}( \bbF )$ by design of MPC \eqref{eq:smpc}.
    Thus, we rely on an inner ellipsoidal approximation of $\bbF$, given by $\mathcal{E} := \{ x \mid \|x - z_{1 \mid k}^* \|_S \leq 1 \}$ for $S \in \R{\n{x} \times \n{x}}$, $S \succ 0$ such that $\mathcal{E}, \subseteq \bbF$.
    Then, by arguments similar to those for proving the multidimensional Chebyshev inequality, we find that
    \begin{equation}\label{eq:theorem2_P_notin_F}
        \bbP[ \| x_{k+1} - z_{1 \mid k}^* \|_S > 1] \leq \Tr(S \Sigma^w).
    \end{equation}
    %%%%%%%%%%%%%%%%%%%%%%% Old arguments that were a little off focus
    % To this end, note that $z_{1 \mid k}^* \in \bbF$ and that \eqref{eq:theorem2_stability_1} already implies nominal stability of MPC \eqref{eq:smpc}.
    % The latter further implies $\| z_{1 \mid k}^* \| \leq \|z_{0 \mid k}^* \|$.
    % Thereby, and exploiting continuity of MPC \eqref{eq:smpc} by design and according to Proposition \eqref{prop:optimal_value_function}, we rely on an inner ellipsoidal approximation of $\bbF$, given by $\mathcal{E} := \{ x \mid \|x - z_{1 \mid k}^* \|_S \leq b \}$ for $S \in \R{\n{x} \times \n{x}}$, $S \succ 0$ and some $b > 0$ such that $\mathcal{E} \subseteq \bbF$.
    % Then, by arguments similar to those for proving the multidimensional Chebyshev inequality, we find that
    % \begin{equation}\label{eq:theorem2_P_notin_F}
    %     \bbP[ \| x_{k+1} - z_{1 \mid k}^* \|_S > b] \leq \frac{\Tr(S \Sigma^w)}{b^2}.
    % \end{equation}
    %% Without loss of generality, $S$ can be scaled appropriately such that $b = 1$.
    % \rev{W.l.o.g.}, $S$ can be scaled appropriately such that $b = 1$.
    % Since $\mathcal{E} \subseteq \bbF$, it holds that $\bbP[x_{k+1} \notin \bbF] \leq \bbP [ \| x_{k+1} - z_{1 \mid k}^* \|_S > b]$ and it follows from \eqref{eq:theorem2_stability_1} and
    %%%%%%%%%%%%%%%%%%%%%%%%
    % \rev{W.l.o.g.}, $S$ can be scaled appropriately such that $b = 1$.
    Since $\mathcal{E} \subseteq \bbF$, it holds that $\bbP[x_{k+1} \notin \bbF] \leq \bbP [ \| x_{k+1} - z_{1 \mid k}^* \|_S > 1]$ and it follows from \eqref{eq:theorem2_stability_1} and \eqref{eq:theorem2_P_notin_F} that
    \begin{multline}
        \E[ \Tilde{V}_N(x_{k+1}) - \Tilde{V}_N(x_k) \mid x_k ] \leq - \kappa(\Tilde{V}_N(x_k)) \\ + \sum_{i=0}^{N-1} \Tr({A^i}^\top Q A^i \Sigma^w) + \lambda \sigma_\mathrm{max}^{\mathcal{X}_0}(S)\Tr(\Sigma^w),
    \end{multline}
    % where $\sigma_\mathrm{max}^{\mathcal{X}_0}(S)$ is the largest singular value of matrices $S$ that occur on $\mathcal{X}_0$ when $b = 1$.
    where $\sigma_\mathrm{max}^{\mathcal{X}_0}(S)$ is the largest singular value of matrices $S$ that occur on $\mathcal{X}_0$.
    Noting that $\varphi(\| w \|_{L^2}) = \sum_{i=0}^{N-1} \Tr({A^i}^\top Q A^i \Sigma^w) + \lambda \sigma_\mathrm{max}^{\mathcal{X}_0}(S)\Tr(\Sigma^w)$ is a class $\mathcal{K}$ function concludes the proof.
\end{proof}
Hence, Theorem \ref{th:issp_controlled_system} shows that system \eqref{eq:system} under MPC law \eqref{eq:mpc_law} is ISSp as long as the system remains in $\mathcal{X}_0$.

\subsection{ISSp under the Combined Control Policy}

In view of Assumption \ref{as:nominal_autonomous_stability}, we employ the control policy \eqref{eq:control_policy_general} with back-up controller $\mu_\mathrm{BackUp}(x) = 0$, i.e.,%\footnote{Different back-up controllers such as a linear-quadratic regulator can be employed as long as it stabilizes system \eqref{eq:system} in the sense of Definition \ref{def:issp}.}
\begin{align}\label{eq:control_policy}
    u_k = \begin{cases} \mu_\mathrm{MPC}(x_k) & \text{if~} x_k \in \mathcal{X}_0 \\ 0 & \text{if~} x_k \notin \mathcal{X}_0 \end{cases}.
\end{align}
Combining the results from Sections \ref{sec:issp_autonomous} and \ref{sec:issp_mpc}, we can conclude about ISSp of system \eqref{eq:system} under control policy \eqref{eq:control_policy}.
To this end, we first define recurrence of a bounded set. % as follows.

\begin{definition}\label{def:recurrence}(From \cite{Culbertson2023}).
    For a bounded set $\mathcal{Y} \subset \R{\n{x}}$, define the hitting time as $\tau_\mathcal{Y} := \inf \{k \in \mathbb{N} \mid x_k \in \mathcal{Y}, x_0 = x \}$. 
    The set $\mathcal{Y}$ is recurrent if $\bbP[\tau_\mathcal{E} < \infty] = 1$ for all $x \in \R{\n{x}}$.
\end{definition}

Loosely speaking, if the set $\mathcal{Y}$ is recurrent for a system, then the system visits $\mathcal{Y}$ in finite time and infinitely often.
Exploiting Theorem \ref{th:issp_autonomous_system}, we can state the following result.

\begin{assumption}\label{as:V_subset_X0}
    There exists $\gamma > \frac{\Tr(P \Sigma^w) \sigma_\mathrm{max}^2 \big( P^\frac{1}{2} \big)}{\sigma_\mathrm{min}^2 \big( (P \! - \! A^\top P A)^\frac{1}{2} \big)}$ such that $\mathcal{V}_\gamma := \{x \in \R{\n{x}} \mid {V}(x) \leq \gamma \} \subset \mathcal{X}_0$.
\end{assumption}

\begin{lemma}\label{lemma:recurrence}
    Under Assumptions \ref{as:disturbance}, \ref{as:nominal_autonomous_stability} and \ref{as:V_subset_X0}, the set $\mathcal{X}_0$ is recurrent for the autonomous system \eqref{eq:disturbed_autonomous_dynamics}.
\end{lemma}

\begin{proof}
    It follows from Theorem 5 in \cite{Culbertson2023}, that $\mathcal{V}_\gamma$ is bounded and recurrent for system \eqref{eq:disturbed_autonomous_dynamics} for $\gamma$ as in Assumption \ref{as:V_subset_X0}. Since $\mathcal{V}_\gamma \subset \mathcal{X}_0$, $\mathcal{X}_0$ is recurrent for system \eqref{eq:disturbed_autonomous_dynamics}.
\end{proof}

Loosely speaking, if system \eqref{eq:system} leaves $\mathcal{X}_0$ at time step $k$, and we stop applying control inputs according to \eqref{eq:control_policy}, the uncontrolled system \eqref{eq:disturbed_autonomous_dynamics} will return to $\mathcal{X}_0$ in finite time.
This enables us to state the final result.

\begin{theorem}\label{th:issp_final}
    Under Assumptions \ref{as:disturbance}, \ref{as:nominal_autonomous_stability} and \ref{as:V_subset_X0}, system \eqref{eq:system} under control policy \eqref{eq:control_policy} is ISSp. 
\end{theorem}

\begin{proof}
    Let $x_k \in \mathcal{X}_0$, which implies according to \eqref{eq:control_policy} the use of the MPC law \eqref{eq:mpc_law}. 
    Note that $\mathcal{X}_0$ is bounded, $0 \in \mathrm{int}(\mathcal{X}_0)$, and $\mathcal{X}_0$ is rendered probabilistically forward invariant under \eqref{eq:mpc_law}.
    Then, from Theorem \ref{th:issp_controlled_system}, it follows that system \eqref{eq:system} under \eqref{eq:control_policy} is ISSp with ISSp Lyapunov function \eqref{eq:lyapunov_candidate} while in $\mathcal{X}_0$.
    % This implies that system \eqref{eq:system} under MPC law \eqref{eq:mpc_law} is ISSp while in $\mathcal{X}_0$.
    If the system leaves $\mathcal{X}_0$, we stop controlling it according to \eqref{eq:control_policy}.
    By Theorem \ref{th:issp_autonomous_system}, the uncontrolled system is ISSp and by Lemma \ref{lemma:recurrence}, the system will return to $\mathcal{X}_0$ in finite time.
\end{proof}

% Loosely speaking, we can intuitively understand the result presented in Theorem \ref{th:issp_final} as follows.
% Let $x_k \in \mathcal{X}_0$, which implies according to \eqref{eq:control_policy} the use of the MPC law \eqref{eq:mpc_law}. 
% Then, from Theorem \ref{th:issp_controlled_system}, it follows that system \eqref{eq:system} under control policy \eqref{eq:control_policy} is ISSp with ISSp Lyapunov function \eqref{eq:lyapunov_candidate} while in $\mathcal{X}_0$.
% If the system leaves $\mathcal{X}_0$, we stop controlling it according to \eqref{eq:control_policy}.
% By Theorem \ref{th:issp_autonomous_system}, the uncontrolled system is ISSp and by Lemma \ref{lemma:recurrence}, the system will return to $\mathcal{X}_0$ in finite time.
% Hence, system \eqref{eq:system} is stabilized under control policy \eqref{eq:control_policy}.

\section{Numeric Example}\label{sec:simulation}

We consider the time-discrete stochastic LTI system
\begin{align}\label{eq:example_system}
\begin{split}
    & x_{k+1} = \begin{bmatrix} 0.924 & -0.100 \\ 0.050 & 1.000 \end{bmatrix} x_k + \begin{bmatrix} 0.025 \\ 0.000 \end{bmatrix} u_k + w_k,  \\
    & w_k \sim \mathcal{N} \left( \begin{bmatrix} 0 \\ 0 \end{bmatrix} , \begin{bmatrix} 0.0050 & 0 \\ 0 & 0.0075 \end{bmatrix} \right)
    % & w_k \sim \mathcal{N} \left( \rev{\begin{bmatrix} 0 & 0 \end{bmatrix}^\top , ~\mathrm{diag}(0.0050, 0.0075)} \right)
\end{split}
\end{align}
with sampling time $T_s = 0.05$.
% with $w_k \sim \mathcal{N} ( \rev{\begin{bmatrix} 0 & 0 \end{bmatrix}^\top , ~\mathrm{diag}(0.0050, 0.0075)} )$ and sampling time $T_s = 0.05$.
% Note that the nominal dynamics of \eqref{eq:example_system} are asymptotically stable.
The nominal dynamics of \eqref{eq:example_system} are asymptotically stable.
We first investigate ISSp of the autonomous system, followed by examining ISSp under stochastic MPC and the control policy \eqref{eq:control_policy}, respectively.

\subsection{ISSp of Autonomous Dynamics}

We find that the function
\begin{equation}\label{eq:sim_example_lyapunov}
    V(z_k) = z_k^\top \begin{bmatrix} 1.093 & 0.554 \\ 0.554 & 2.915 \end{bmatrix} z_k
\end{equation}
is a Lyapunov function for the autonomous nominal dynamics of \eqref{eq:example_system}.
Hence, Assumption \ref{as:nominal_autonomous_stability} is satisfied and the autonomous uncertain dynamics are ISSp by Theorem \ref{th:issp_autonomous_system}.
%%%%%%%% Ein Reviewer hat vorgeschlagen, den folgenden Abshcnitt aufgrund seines geringen Mehrwerts für das Paper raus zu nehmen:
We show 100 state trajectories of the autonomous uncertain system when (i) starting from the same initial condition $x_0 = \begin{bmatrix} 10 & 0 \end{bmatrix}^\top$ in Figure \ref{fig:autonomous_system__state}, left column, and (ii) starting from different initial conditions in Figure \ref{fig:autonomous_system__state}, right column.
In both cases, we observe convergence (in probability) to a neighborhood of the origin according to Definition \ref{def:issp}.
Figure \ref{fig:ISSp_lyapunov_autonomous} shows \eqref{eq:sim_example_lyapunov} evaluated along the 100 state trajectories of the autonomous uncertain system, as shown in Figure \ref{fig:autonomous_system__state}, left column.
The horizontal black line indicates the value of $\eqref{eq:sim_example_lyapunov}$ for which we cannot expect \eqref{eq:sim_example_lyapunov} to further decrease since $- \frac{\sigma_\mathrm{min}^2 \big((P \! - \! A^\top P A)^{\frac{1}{2}} \big)}{\sigma_\mathrm{max}^2 \big( P^\frac{1}{2} \big)} {V}(x_k) \ + \! \Tr(P \Sigma^w) \! = \! 0$, see Theorem \ref{th:issp_autonomous_system}. 

\begin{figure}[t]
    \centering
    \includegraphics[width=\linewidth]{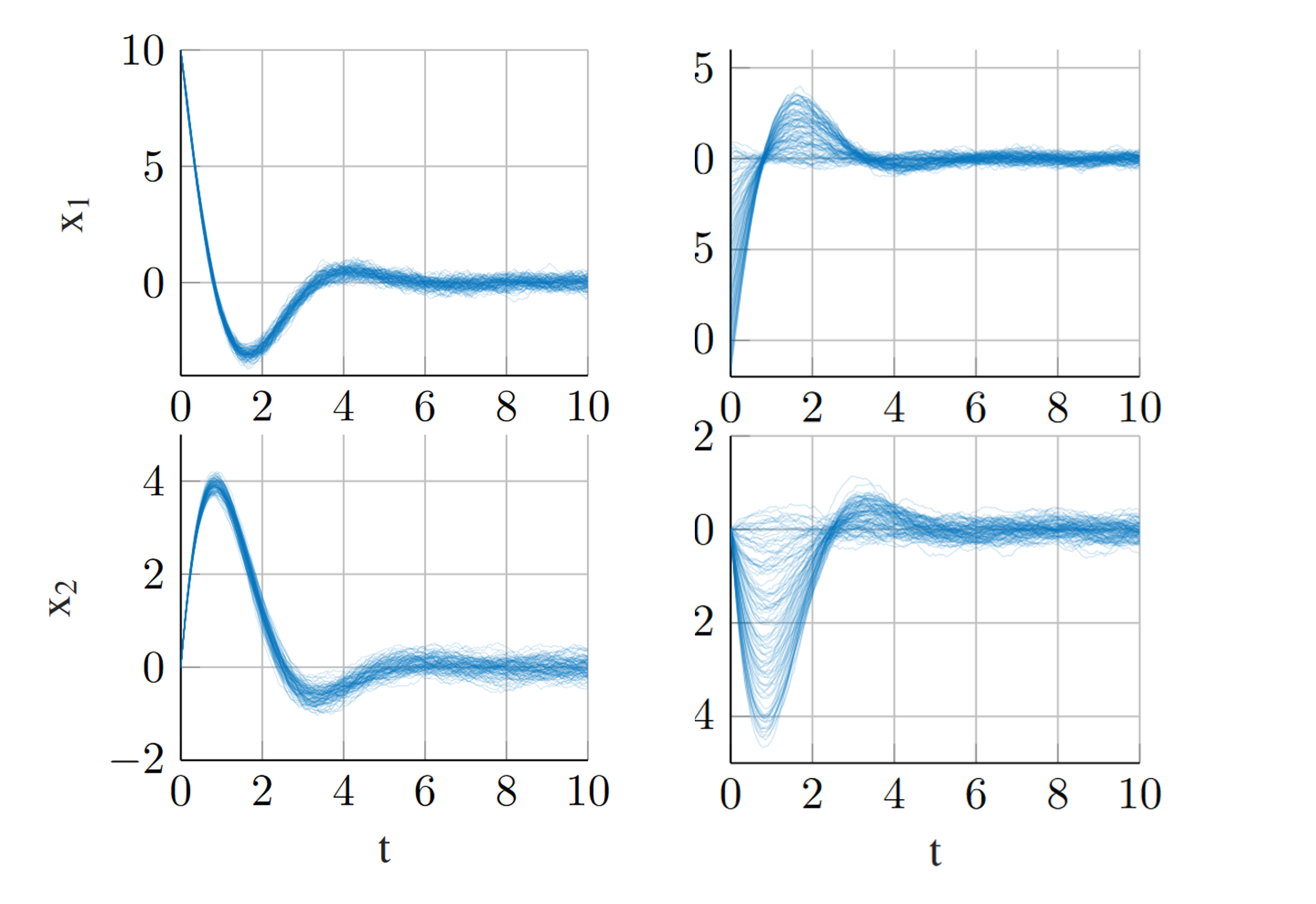}
    \caption{State trajectories of the uncertain autonomous system. Left column: 100 sample trajectories starting at the same initial state. Right column: 100 sample trajectories starting at different initial states.}
    \label{fig:autonomous_system__state}
\end{figure}

\begin{figure}[t]
    \centering
    \input{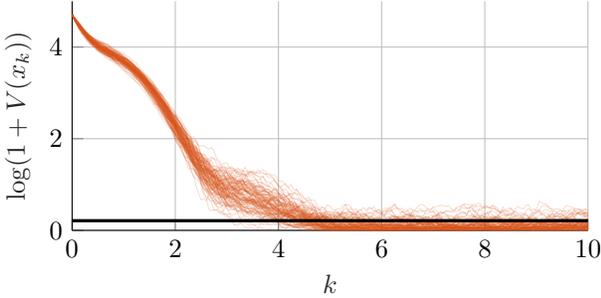}
    \caption{Evolution of \eqref{eq:sim_example_lyapunov} along the 100 state trajectories of the uncertain autonomous system shown in Figure \ref{fig:autonomous_system__state}, left column.}
    \label{fig:ISSp_lyapunov_autonomous}
    % \vspace{-10pt}
\end{figure}
%%%%%%%%

\subsection{ISSp in Closed-Loop}

Next, we design an MPC \eqref{eq:smpc} that satisfies Assumptions \ref{as:quadratic_cost} and \ref{as:standard_stability}.
For the quadratic state and terminal cost, we choose
\begin{equation}
    Q \! = \! \begin{bmatrix} 2 & 0 \\ 0 & 0.1 \end{bmatrix}, R \! = \! 1, \text{and~} Q_f \! = \! \begin{bmatrix} 14.250 & 1.213 \\ 1.213 & 28.339 \end{bmatrix},
\end{equation}
where $Q_f$ is computed via the associated infinite-horizon linear-quadratic regulator problem.
% \rev{$Q = \mathrm{diag}(2, 0.1)$, $R = 1$, and compute $Q_f$ via the associated infinite-horizon linear-quadratic regulator problem.}
The latter is in addition used to define the terminal controller, and the terminal region is chosen to be a sublevel set of the terminal cost function.
It is easy to verify that this design satisfies Assumption \ref{as:standard_stability}.
The constraints are given via $-1 \leq x_1 \leq 12$, $-2 \leq x_2 \leq 4$ and $-37 \leq u \leq 37$, implicitly defining \eqref{eq:state_constraints} and \eqref{eq:input_constraints}.
We compute the tightened constraints \eqref{eq:tightened_constraints} using $\delta = 0.15$ and exploiting the known Gaussian distribution of $w_k$.
% It can be verified that Assumption \ref{as:V_subset_X0} is satisfied for the described set-up.
Assumption \ref{as:V_subset_X0} is satisfied for the described set-up.

Figure \ref{fig:closed_loop_states} shows the results of 100 runs of system \eqref{eq:example_system} under control policy \eqref{eq:control_policy}.
We observe probabilistic convergence to a neighborhood of the origin according to Definition \ref{def:issp}.
While a large fraction of the trajectories remains safe, a small fraction leaves the safe set, as best seen in the phase plane plot.
Whenever this is the case, we stop applying control inputs according to the back-up strategy.
All trajectories that leave the safe set return to it in finite time, illustrating its recurrent nature according to Lemma \ref{lemma:recurrence}.
% Using a different back-up controller, it is possible to make the unsafe trajectories return faster to the safe set.
% However, the satisfaction of input constraints can then not be guarantees anymore.

\begin{figure}[b!]
    \centering
    \includegraphics[width=\linewidth]{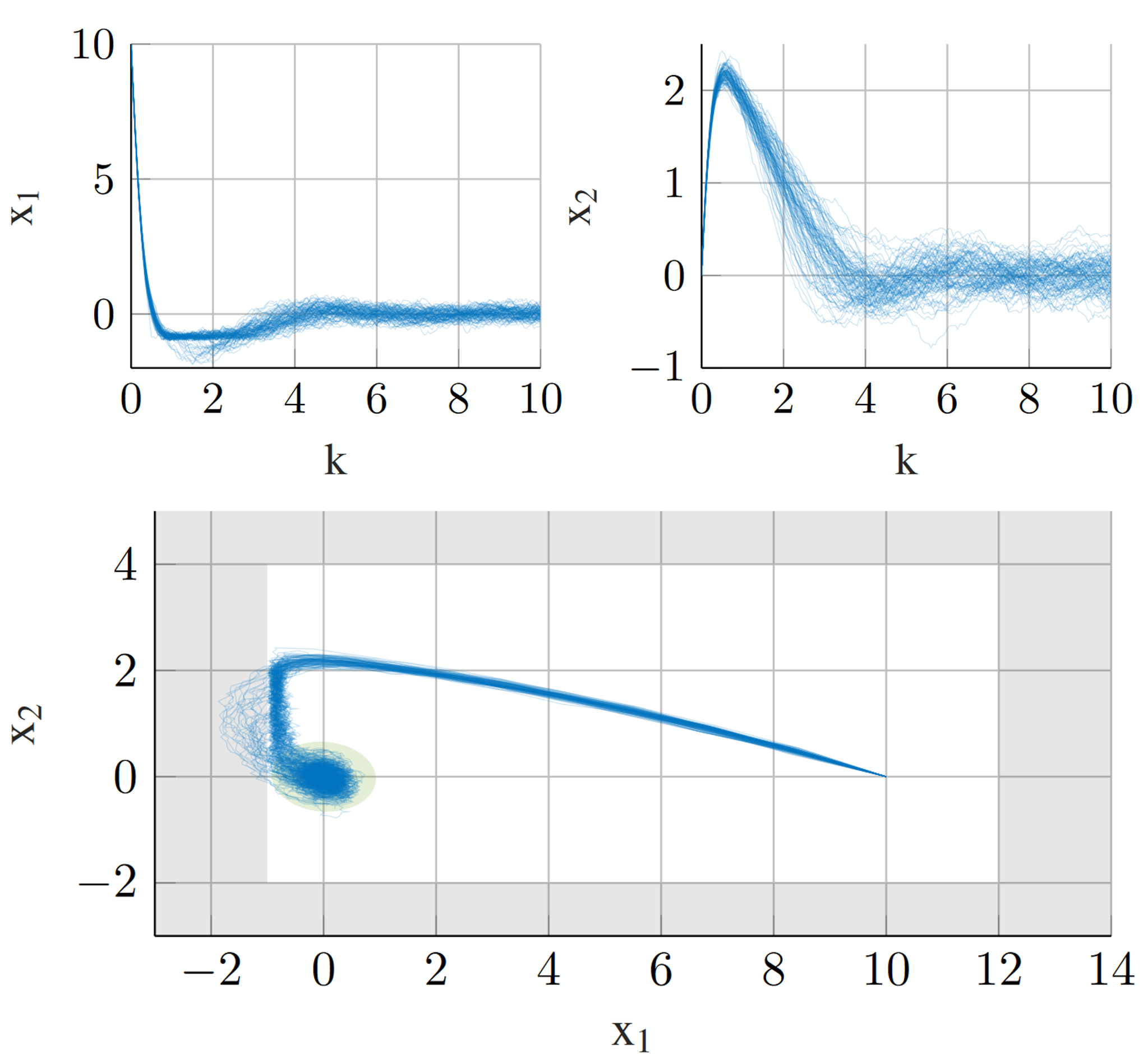}
    \caption{Top: State trajectories of system \eqref{eq:example_system} in closed-loop over time. Bottom: Corresponding phase plane plot. The gray-shaded are indicates the state constraints, while the green-shaded area is the terminal region of MPC \eqref{eq:smpc}.}
    \label{fig:closed_loop_states}
    % \vspace{-10pt}
\end{figure}
\section{Conclusions}\label{sec:conclusions}

\addtolength{\textheight}{-11.5cm}

We have investigated stability properties of stochastic MPC.
In particular, we have provided conditions for which stochastic MPC is guaranteed to be input-to-stable in probability.
In contrary to traditional stability analysis, the proposed results do not rely on strict repeated feasibility of the MPC.
Loss of feasibility is captured by a back-up controller, which is explicitly accounted for in the stability analysis.
Combining both the stability results for the MPC and the back-up controller, we have guaranteed input-to-state-stability in probability in closed-loop. 

Future research will be dedicated towards deriving rigorous probabilistic safety guarantees using different back-up control strategies.
Furthermore, the connection between traditional input-to-state stability for bounded disturbances and its probabilistic counterpart as well as implications on the connection between robust and stochastic MPC formulations will be investigated.


%%%%%%%%%%%%%%%%%%%%%%%%%%%%%%%%%%%%%%%%%%%%%%%%%%%%%%%%%%%%%%%%%%%%%%%%%%%%%%%%

%%%%%%%%%%%%%%%%%%%%%%%%%%%%%%%%%%%%%%%%%%%%%%%%%%%%%%%%%%%%%%%%%%%%%%%%%%%%%%%%


% \section*{ACKNOWLEDGMENT}

% The preferred spelling of the word ÒacknowledgmentÓ in America is without an ÒeÓ after the ÒgÓ. Avoid the stilted expression, ÒOne of us (R. B. G.) thanks . . .Ó  Instead, try ÒR. B. G. thanksÓ. Put sponsor acknowledgments in the unnumbered footnote on the first page.

%%%%%%%%%%%%%%%%%%%%%%%%%%%%%%%%%%%%%%%%%%%%%%%%%%%%%%%%%%%%%%%%%%%%%%%%%%%%%%%%

\bibliographystyle{IEEEtran}
\bibliography{Bibliography}

\end{document}